\documentclass[11pt]{article}
\usepackage{amssymb, amsmath, color, tikz, subcaption, enumerate, multirow, array, amsthm}
\newtheorem{definition}{Definition}[section]
\newtheorem{theorem}{Theorem}
\usepackage[margin=0.95in]{geometry}
\usepackage{charter,euler}

\newcommand{\R}{\mathbb{R}}

\newcommand{\G}{\mathcal{G}}

\newcommand{\ggd}{\mathrm{GGD}}

\newcommand{\gmd}{\mathrm{GMD}}
\newcommand{\cost}{\mathrm{Cost}}

\newcommand{\len}[1]{{\lvert\overline{#1}\rvert}}
\newcommand{\eqdef}{\stackrel{\text{def}}{=}}

\newcommand{\thmref}[1]{Theorem~\ref{thm:#1}}
\newcommand{\figref}[1]{Fig.~\ref{fig:#1}}
\newcommand{\secref}[1]{Section~\ref{sec:#1}}
\newcommand{\eqnref}[1]{(\ref{eqn:#1})}

\newcommand{\tabref}[1]{Table~\ref{table:#1}}

\title{Graph Mover's Distance: An Efficiently Computable Distance Measure for Geometric Graphs}
\date{}
\author{
    Sushovan Majhi\thanks{School of Information, 
    University of California, Berkeley, USA, 
    \tt{smajhi@berkeley.edu}}}
\index{Majhi, Sushovan}
    
\begin{document}
\thispagestyle{empty}
\maketitle

\begin{abstract}
Many applications in pattern recognition represent patterns as a geometric
graph. The geometric graph distance (GGD) has recently been studied in
\cite{majhi2022distance} as a meaningful measure of similarity between two
geometric graphs. Since computing the GGD is known to be $\mathcal{NP}$--hard,
the distance measure proves an impractical choice for applications. As a
computationally tractable alternative, we propose in this paper the Graph
Mover's Distance (GMD), which has been formulated as an instance of the earth
mover's distance. The computation of the GMD between two geometric graphs with
at most $n$ vertices takes only $O(n^3)$-time. Alongside studying the metric
properties of the GMD, we investigate the stability of the GGD and GMD. The GMD
also demonstrates extremely promising empirical evidence at recognizing letter
drawings from the {\tt LETTER} dataset \cite{da_vitoria_lobo_iam_2008}. 
\end{abstract}

\section{Introduction}
Graphs have been a widely accepted object for providing structural
representation of patterns involving relational properties. While hierarchical
patterns are commonly reduced to a string \cite{FU1971155} or a tree
representation \cite{1672247}, non-hierarchical patterns generally require a
graph representation. The problem of pattern recognition in such a
representation then requires quantifying (dis-)similarity between a query graph
and a model or prototype graph.
Defining a relevant distance measure for a class of graphs has been studied for
almost five decades now and has a myriad of applications including chemical
structure matching \cite{willett_similarity_1994}, fingerprint matching
\cite{raymond_effectiveness_2002}, face identification \cite{954601}, and symbol
recognition \cite{954603}.

Depending on the class of graphs of interest and the area of application,
several methods have been proposed. Graph isomorphisms
\cite{10.1145/321556.321562} or subgraph isomorphisms can be considered. These,
however, cannot cope with (sometimes minor) local and structural deformations of
the two graphs. To address this issue, several alternative distance measures
have been studied. We particularly mention \emph{edit distance}
\cite{sanfeliu_distance_1983,justice_binary_2006} and \emph{inexact matching
distance} \cite{bunke_inexact_1983}. Although these distance measures have been
battle-proven for attributed graphs (i.e., combinatorial graphs with finite
label sets), the formulations seem inadequate in providing meaningful similarity
measures for geometric graphs. 

A geometric graph  belongs to a special class of attributed graphs having an
embedding into a Euclidean space $\R^d$, where the vertex labels are
inferred from the Euclidean locations of the vertices and the edge labels are the Euclidean lengths of
the edges.

In the last decade, there has been a gain in practical applications involving
comparison of geometric graphs, such as road-network or map comparison
\cite{akpw-mca-15}, detection of chemical structures using their spatial bonding
geometry, etc. In addition, large datasets like \cite{da_vitoria_lobo_iam_2008}
are being curated by pattern recognition and machine learning communities.

\subsection{Related Work and Our Contribution}
We are inspired by the recently developed geometric graph distance (GGD) in
\cite{cgkss-msgg-09,majhi2022distance}. Although the GGD succeeds to be a
relevant distance measure for geometric graphs, its computation, unfortunately,
is known to be $\mathcal{NP}$-hard. Our motivation stems from applications that
demand an efficiently computable measure of similarity for geometric graphs. The
formulation of our graph mover's distance is based on the theoretical
underpinning of the GGD. The GMD provides a meaningful yet computationally
efficient similarity measure between two geometric graphs. 

In \secref{ggd}, we revisit the definition of the (GGD) to investigate its
stability under Hausdorff perturbation. \secref{gmd} is devoted to the study of
the GMD. The GMD has been shown to render a \emph{pseudo}-metric on the class of
(ordered) geometric graphs. Finally, we apply the GMD to classify letter
drawings in \secref{exp}. Our experiment involves matching each of $2250$ test
drawings, modeled as geometric graphs, to $15$ prototype letters from the
English alphabet. For the drawings with {\tt LOW} distortion, the correct letter
has been found among the top $3$ matches at a rate of $98.93\%$, where the
benchmark accuracy is $99.6\%$ obtained using a $k$-nearest neighbor classifier
($k$-NN) with the graph edit distance \cite{bunke_inexact_1983}.

\section{Geometric Graph Distance (GGD)}\label{sec:ggd} We first formally define
a geometric graph. Throughout the paper, the dimension of the ambient Euclidean
space is denoted by $d\geq1$. We also assume that the cost coefficients $C_V$
and $C_E$ are positive constants.
\begin{definition}[Geometric Graph]\label{def:graph} A \emph{geometric graph} of
$\R^d$ is a (finite) combinatorial graph $G=(V^G,E^G)$ with vertex set
$V^G\subset\R^d$, and the Euclidean straight-line segments
$\{\overline{ab}\mid (a,b)\in E^G\}$ intersect (possibly) at their
endpoints.
\end{definition}
We denote the set of all geometric graphs of $\R^d$ by $\G(\R^d)$. Two geometric
graphs $G=(V^G,E^G)$ and $H=(V^H,E^H)$ are said to be \emph{equal}, written
$G=H$, if and only if $V^G=V^H$ and $E^G=E^H$. We make no distinction between a
geometric graph $G=(V^G,E^G)$ and its \emph{geometric realization} as a subset
of $\R^d$; an edge $(u,v)\in E^G$ can be identified as the line-segment
$\overline{uv}$ in $\R^d$, and its length by the Euclidean length $\len{uv}$. 

Following the style of \cite{majhi2022distance}, we first revisit the definition
of GGD. The definition uses the notion of an inexact matching. In order to
denote a deleted vertex and a deleted edge, we introduce the \emph{dummy vertex}
$\epsilon_V$ and the \emph{dummy edge} $\epsilon_E$, respectively. 
\begin{definition}[Inexact Matching]\label{def:pi} Let $G,H\in\G(\R^d)$ be two
geometric graphs. A relation
$\pi\subseteq(V^G\cup\{\epsilon_V\})\times(V^H\cup\{\epsilon_V\})$ is called an
(inexact) matching if for any $u\in V^G$ (resp.~$v\in V^H$) there is exactly one
$v\in V^H\cup\{\epsilon_V\}$ (resp.~$u\in V^G\cup\{\epsilon_V\}$) such that
$(u,v)\in\pi$.
\end{definition}   

The set of all matchings between graphs $G,H$ is denoted by $\Pi(G,H)$.
Intuitively, a matching $\pi$ is a relation that covers the vertex sets
$V^G,V^H$ exactly once. As a result, when restricted to $V^G$ (resp.~$V^H$), a
matching $\pi$ can be expressed as a map $\pi:V^G\to V^H\cup\{\epsilon_V\}$
(resp.~$\pi^{-1}:V^H\to V^G\cup\{\epsilon_V\}$). In other words, when
$(u,v)\in\pi$ and $u\neq\epsilon_V$ (resp.~$v\neq\epsilon_V$), it is justified
to write $\pi(u)=v$ (resp.~$\pi^{-1}(v)=u$). It is evident from the definition
that the induced map $$\pi:\{u\in V^G\mid\pi(u)\neq\epsilon_V\}\to \{v\in
V^H\mid\pi^{-1}(v)\neq\epsilon_V\}$$ is a bijection. For edges $e=(u_1,u_2)\in
E^G$ and $f=(v_1,v_2)\in E^H$, we introduce the short-hand
$\pi(e):=(\pi(u_1),\pi(u_2))$ and $\pi^{-1}(f):=(\pi^{-1}(v_1),\pi^{-1}(v_2))$.

Another perspective of $\pi$ is to view it as a matching between portions of $G$
and $H$, (possibly) after applying some edits on the two graphs. For example,
$\pi(u)=\epsilon_V$ (resp.~$\pi^{-1}(v)=\epsilon_V$) encodes deletion of the
vertex $u$ from $G$ (resp.~$v$ from $H$), whereas $\pi(e)=\epsilon_E$
(resp.~$\pi^{-1}(f)=\epsilon_E$) encodes deletion of the edge $e$ from $G$
(resp.~$f$ from $H$). Once the above deletion operations have been performed on
the graphs, the resulting subgraphs of $G$ and $H$ become isomorphic, which are
finally matched by translating the remaining vertices $u$ to $\pi(u)$. Now, the
cost of the matching $\pi$ is defined as the total cost for all of these
operations:
\begin{definition}[Cost of a Matching] \label{def:split-ggd}Let $G,H\in\mathcal
G(\R^d)$ be geometric graphs and $\pi\in\Pi(G,H)$ an inexact matching. The cost
of $\pi$, is $\cost(\pi)=$
\begin{equation}\label{eqn:split-ggd}
\begin{split}
&\underbrace{\sum_{\substack{u\in V^G \\ \pi(u)\neq\epsilon_V}} 
C_V|u-\pi(u)|}_\text{vertex translations} + 
\underbrace{\sum_{\substack{e\in E^G \\ \pi(e)\neq\epsilon_E}} 
C_E\big||e|-|\pi(e)|\big|}_\text{edge translations}+ \quad\quad\underbrace{\sum_{\substack{e\in E^G \\ \pi(e)=\epsilon_E}} C_E|e|}
_\text{edge deletions} + \underbrace{\sum_{\substack{f\in E^H \\ \pi^{-1}(f)=\epsilon_E}} C_E|f|}_\text{edge deletions}.        
\end{split}
\end{equation} 
\end{definition}

\begin{definition}[$\ggd$]\label{def:ggd} For geometric graphs $G,H\in\G(\R^d)$,
their geometric graph distance, $\ggd(G,H)$, is
\[\ggd(G,H)\eqdef\min_{\pi\in\Pi(G,H)}\cost(\pi)\;.\]
\end{definition}

\subsection{Stability of GGD}
A distance measure is said to be \emph{stable} if it does not change much if the
inputs are \emph{perturbed} only slightly. Usually, the change is expected to be
bounded above by the amount of perturbation inflicted on the inputs. The
perturbation is measured under a suitable choice of metric. In the context of
geometric graphs, it is natural to wonder if the GGD is stable under the
Hausdorff distance between two graphs. To our disappointment, we can easily see
for the graphs shown in \figref{stable} that the GGD is positive, whereas the
Hausdorff distance between their realizations is zero. So, the Hausdorff
distance between the graphs can not bound their GGD from above.
\begin{figure}[tbh]
    \centering
    \begin{tikzpicture}[scale=0.8]
    \foreach \i in {-1,...,5} { \draw [very thin, gray, |-|] (\i,0) -- (\i+1,0); }
    \foreach \i in {-1,...,5} { \draw [very thin, gray, |-|] (\i,1) -- (\i+1,1); }
    \filldraw[thick] (0,0) circle (2pt) node[above] {$v_1$} -- (4,0) circle (2pt) 
    node[above] {$v_2$}; 
    \filldraw[thick] (0,1) circle (2pt) node[above] {$u_1$} -- (3,1) circle (2pt) 
    node[above] {$u_2$} -- (4,1) circle (2pt) node[above] {$u_3$}; 
    \node[below] at (2.5,0) {$H$};
    \node[below] at (1.5,1) {$G$}; 
    \end{tikzpicture}
    \caption{The graphs $G$ (top) and $H$ (bottom) are embedded in the real line;
    the distance between consecutive ticks is $1$ unit. The Hausdorff distance
    between $G$ and $H$ is zero, however $\ggd(G,H)=C_V + C_E$ is non-zero. The
    optimal matching is given by $\pi(u_1)=v_1$, $\pi(u_2)=v_2$, and
    $\pi(u_3)=\epsilon_V$.}
    \label{fig:stable}
\end{figure}
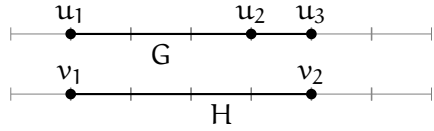
    
One might think that the $\ggd$ is stable when the Hausdorff distance only
between the vertices is considered. However, the graphs in \figref{stable-1}
indicate otherwise.
\begin{figure}[thb]
\centering 
\begin{tikzpicture}[scale=0.75]
\foreach \i in {0,...,3} { \draw [very thin, gray] (\i-4,0) -- (\i-4,3)  node
[below] at (\i-4,0) {\footnotesize$\i$}; } \foreach \i in {0,...,3} { \draw
[very thin, gray] (-4,\i) -- (-1,\i) node [left] at (-4,\i) {\footnotesize$\i$};
} \filldraw (-4,2) circle (2pt) node[anchor=south west] {$u_1$} -- (-2,0) circle
(2pt) node[anchor=south west] {$u_3$} -- (-4,0) circle (2pt) node[anchor=south
west] {$u_2$}; \foreach \i in {0,...,3} { \draw [very thin, gray] (\i+1,0) --
(\i+1,3)  node [below] at (\i+1,0) {\footnotesize$\i$}; }
\foreach \i in {0,...,3} { \draw [very thin, gray] (1,\i) -- (4,\i) node [left]
at (1,\i) {\footnotesize$\i$}; }
\filldraw (3,0) circle (2pt) node[anchor=south west] {$v_3$} -- (1,2) circle
(2pt) node[anchor=south west] {$v_1$} -- (1,0) circle (2pt) node[anchor=south
west] {$v_2$};
\end{tikzpicture}
\caption{For the graphs $G,H\in\G(\R^2)$, the Hausdorff distance between the
vertex sets is zero, however $\ggd(G,H)=4C_E$ is non-zero. The optimal matching
is given by $\pi(u_1)=v_1$, $\pi(u_3)=v_3$, $\pi(u_2)=\epsilon_V$, and
$\pi^{-1}(v_2)=\epsilon_V$.}
\label{fig:stable-1}
\end{figure}
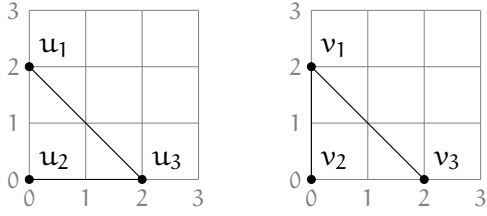

Under strong requirements, however, it is not difficult to prove the following
result on the stability of $\ggd$ under the Hausdorff distance.
\begin{theorem}[Hausdorff Stability of GGD]\label{thm:ggd}
Let $G,H\in\G(\R^d)$ be geometric graphs with a graph isomorphism $\pi:V^G\to
V^H$. If $\delta>0$ is such that $|u-\pi(u)|\leq\delta$ for all $u\in V^G$, then
\[ \ggd(G,H)\leq C_V|V^G|\delta. \]
\end{theorem}
\begin{proof}
The given graph isomorphism $\pi$ is a bijective mapping between the vertices of
$G$ and $H$. So, $\pi\in\Pi(G,H)$, i.e., it defines an inexact matching. Since
$\pi$ is a graph isomorphism, it does not delete any vertex or edge. More
formally, for all $u\in V^G$ and $v\in V^H$, we have $\pi(u)\neq\epsilon_V$ and
$\pi^{-1}(v)\neq\epsilon_V$, respectively. Also, for all $e\in E^G$ and $f\in
E^H$, we have $\pi(e)\neq\epsilon_E$ and $\pi^{-1}(f)\neq\epsilon_E$,
respectively. From \eqnref{split-ggd}, the cost 
\[
\cost(\pi)=\sum_{u\in V^G} C_V|u-\pi(u)|\leq C_V|V^G|\delta.
\]
So, $\ggd(G,H)\leq\cost(\pi)\leq C_V|V^G|\delta$.
\end{proof}

\section{Graph Mover's Distance (GMD)}\label{sec:gmd} We define the \emph{Graph
Mover's Distance} for two ordered geometric graphs. A geometric graph is called
\emph{ordered} if its vertices are ordered or indexed. In that case, we denote
the vertex set as a (finite) sequence $V^G=\{u_i\}_{i=1}^m$. Let us denote by
$\G^O(\R^d)$ the set of all ordered geometric graphs of $\R^d$. The formulation
of the GMD uses the framework known as the earth mover's distance (EMD).

\subsection{Earth Mover's Distance (EMD)}
The EMD is a well-studied distance measure between weighted point sets, with
many successful applications in a variety of domains; for example, see
\cite{hargreaves_earth_2020,kusner_word_2015,ren_robust_2011,rubner_earth_2000}.
The idea of the EMD was first conceived by Monge \cite{monge_memoire_1781} in
1781, in the context of transportation theory. The name ``earth mover's
distance'' was coined only recently, and is well-justified due to the following
analogy. The first weighted point set can be thought of as piles of earth (dirt)
lying on the point sites, with the weight of a site indicating the amount of
earth; whereas, the other point set as pits of volumes given by the
corresponding weights. Given that the total amount of earth in the piles equals
the total volume of the pits, the EMD computes the least (cumulative) cost
needed to fill all the pits with earth. Here, a unit of cost corresponds to
moving a unit of earth by a unit of ``ground distance'' between the pile and the
pit.

The EMD can be cast as a transportation problem on a bipartite graph, which has
several efficient implementations, e.g., the network simplex algorithm
\cite{ahuja_network_2013,pele_linear_2008}. Let the weighted point sets
$P=\{(p_i, w_{p_i})\}_{i=1}^m$ and $Q=\{(q_j, w_{q_j})\}_{j=1}^n$ be a set of
suppliers and a set of consumers, respectively. The weight $w_{p_i}$ denotes the
total supply of the supplier $p_i$, and $w_{q_j}$ the total demand of the
consumer $q_j$. The matrix $[d_{i,j}]$ is the matrix of ground distances, where
$d_{i,j}$ denotes the cost of transporting a unit of supply from $p_i$ to $q_j$.
We also assume the \emph{feasibility condition} that the total supply equals the
total demand: 
\begin{equation}\label{eqn:feasibility}
    \sum_{i=1}^m w_{p_i} = \sum_{j=1}^n w_{q_j}\;.
\end{equation}
A \emph{flow} of supply is given by a matrix $[f_{i,j}]$ with $f_{i,j}$ denoting
the units of supply transported from $p_i$ to $q_j$. We want to find a flow that
minimizes the overall cost
\[
    \sum_{i=1}^m\sum_{j=1}^n f_{i,j}d_{i,j}
\]
subject to:
\begin{align}
    &f_{i,j}\geq 0\text{ for any }i=1,\ldots,m\text{ and }j=1,\dots,n \label{eqn:emd-1} \\
    &\sum_{j=1}^n f_{i,j} = w_i\text{ for any }i=1,\ldots,m \label{eqn:emd-2}\\
    &\sum_{i=1}^m  f_{i,j} = w_j\text{ for any }j=1,\ldots,n \label{eqn:emd-3},
\end{align}
Constraint \eqnref{emd-1} ensures a flow of units from $P$ to $Q$, and not vice
versa; constraint \eqnref{emd-2} dictates that a supplier must send all its
supply---not more or less; constraint \eqnref{emd-3} guarantees that the demand
of every consumer is exactly fulfilled. 

The \emph{earth mover's distance} (EMD) is then defined by the cost of the
optimal flow. A solution always exists, provided condition \eqnref{feasibility}
is satisfied. The weights and the ground distances can be chosen to be any
non-negative numbers. However, we choose them appropriately in order to solve
our graph matching problem.
\begin{figure}[thb]
\centering
\begin{tikzpicture}[scale=1]
\filldraw[gray, shorten >=0.4cm, ->] (-2,2) -- (2, 2);
\filldraw[gray, shorten >=0.4cm, ->] (-2,2) -- (2, 0.5);
\filldraw[red, shorten >=0.4cm, ->] (-2,2) -- 
node[near end, fill=white, inner sep = 1] {$1$} (2, -1);
\filldraw[gray, shorten >=0.3cm, ->] (-2,1) -- (2, 2);
\filldraw[red, shorten >=0.3cm, ->] (-2,1) -- 
node[near end, fill=white, inner sep = 1] {$1$} (2, 0.5);
\filldraw[gray, shorten >=0.3cm, ->] (-2,1) -- (2, -1);
\filldraw[red, shorten >=0.2cm, ->] (-2,0) -- 
node[near end, fill=white, inner sep = 1] {$1$} (2, 2);
\filldraw[gray, shorten >=0.2cm, ->] (-2,0) -- (2, 0.5);
\filldraw[gray, shorten >=0.2cm, ->] (-2,0) -- (2, -1);
\filldraw[gray, shorten >=0.1cm, ->] (-2,-1) -- (2, 2);
\filldraw[gray, shorten >=0.1cm, ->] (-2,-1) -- (2, 0.5);
\filldraw[red, shorten >=0.1cm, ->] (-2,-1) -- 
node[near end, fill=white, inner sep = 1] {$2$} (2, -1);

\fill[] (-2,2) circle (2pt) node[anchor = east] {$u_1$} node[anchor = north] {$1$};
\fill[] (-2,1) circle (2pt) node[anchor = east] {$u_2$} node[anchor = north] {$1$};
\fill[] (-2,0) circle (2pt) node[anchor = east] {$u_3$} node[anchor = north] {$1$};
\fill[gray] (-2,-1) circle (2pt) node[anchor = east] {$u_4$} node[anchor = north] {$2$};
\fill[] (2,2) circle (2pt) node[anchor = west] {$v_1$} node[anchor = north] {$1$};
\fill[] (2, 0.5) circle (2pt) node[anchor = west] {$v_2$} node[anchor = north] {$1$};
\fill[gray] (2,-1) circle (2pt) node[anchor = west] {$v_3$} node[anchor = north] {$3$};

\end{tikzpicture}
\caption{The bipartite network used by the $\gmd$ is shown for two ordered
graphs $G,H$ with vertex sets $V^G=\{u_1,u_2,u_3\}$ and $V^H=\{v_1,v_2\}$,
respectively. The dummy nodes $u_4$ for $G$ and $v_3$ for $H$, respectively,
have been shown in gray. Below each node, the corresponding weights are shown. A
particular flow has been depicted here. The gray edges do not transport
anything. A red edge has a non-zero flow with the transported units shown on
them.}
\label{fig:gmd}
\end{figure}
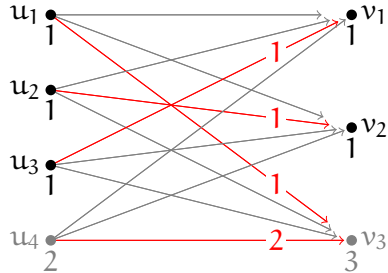

\subsection{Defining the GMD}
Let $G,H\in\G^O(\R^d)$ be two ordered geometric graphs of $\R^d$ with
$V^G=\{u_i\}_{i=1}^m$ and $V^H=\{v_j\}_{j=1}^n$. For each $i=1,\ldots,m$, let
$E^G_i$ denote the (row) $m$--vector containing the lengths of (ordered) edges
incident to the vertex $u_i$ of $G$. More precisely, the
\[
k\text{th element of }E^G_i = \begin{cases}|e^G_{i,k}|,\text{ if }e^G_{i,k}:=(u_i, u_k)\in E^G\\
0,\text{ otherwise.}    
\end{cases}
\]
Similarly, for each $j=1,\ldots,n$, we define $E^H_j$ to be the (row) $n$--vector with
the 
\[
k\text{th element of }E^H_j = \begin{cases}|e^H_{j,k}|,\text{ if }e^H_{j,k}:=(v_j, v_k)\in E^H\\
0,\text{ otherwise.}
\end{cases}
\]
In order to formulate the desired instance of the EMD, we take the point sets to
be $P=\{u_i\}_{i=1}^{m+1}$ and $Q=\{v_j\}_{j=1}^{n+1}$. Here, $u_{m+1}$ and
$v_{n+1}$ have been taken to be a dummy supplier and dummy consumer,
respectively, to incorporate vertex deletion into our GMD framework. The weights
on the sites are defined as follows:
\[
    w_{u_i}=1\text{ for }i=1\ldots,m\text{ and }w_{u_{m+1}}=n\;.
\]
And,
\[
    w_{v_j}=1\text{ for }j=1\ldots,n\text{ and }w_{v_{m+1}}=m\;.   
\]
We note that the feasibility condition \eqnref{feasibility} is satisfied: $m+n$
is the total weight for both $P$ and $Q$. An instance of the transportation
problem is depicted in \figref{gmd}. 

Finally, the ground distance from $u_i$ to $v_j$ is defined by:
\[
    d_{i,j}=\begin{cases}C_V|u_i-v_j| 
    +&C_E\|E^G_iD_{m\times p}-E^H_jD_{n\times p}\|_1,\\
    &\text{ if }1\leq i\leq m,1\leq j\leq n\\
    C_E\|E^H_j\|_1,&\text{ if }i=m+1\text{ and } 1\leq j\leq n \\
    C_E\|E^G_i\|_1,&\text{ if }1\leq i\leq m\text{ and }j=n+1\\
    0,&\text{ otherwise}.
    \end{cases}
\]
Here, $p=\min\{m,n\}$, the $1$--norm of a row vector is denoted by
$\|\cdot\|_1$, and $D$ denotes a diagonal matrix with the all diagonal entries
being $1$. 

\begin{figure}[thb]
\centering 
\begin{tikzpicture}[scale=1]
\foreach \i in {0,...,2} { \draw [very thin, gray] (\i-4,0) -- (\i-4,2)  node
[below] at (\i-4,0) {\footnotesize$\i$}; } \foreach \i in {0,...,2} { \draw
[very thin, gray] (-4,\i) -- (-2,\i) node [left] at (-4,\i) {\footnotesize$\i$};
} 

\filldraw (-4,2) circle (2pt) node[anchor=south west] {$u_4$} -- (-2,2) circle
(2pt) node[anchor=south west] {$u_1$} -- (-3,1) circle (2pt)
node[anchor=west] {$u_5$} -- (-4,0) circle (2pt) node[anchor=south] {$u_2$} --
(-2,0) circle (2pt) node[anchor=south west] {$u_3$}; 

\foreach \i in {0,...,2} { \draw [very thin, gray] (\i+1,0) -- (\i+1,2)  node
[below] at (\i+1,0) {\footnotesize$\i$}; }
\foreach \i in {0,...,2} { \draw [very thin, gray] (1,\i) -- (3,\i) node [left]
at (1,\i) {\footnotesize$\i$}; }
\filldraw (1,2) circle (2pt) node[anchor=south west] {$v_3$} -- (1,0) circle
(2pt) node[anchor=west] {$v_1$} -- (2,1) circle (2pt)
node[anchor=west] {$v_5$} -- (3,2) circle (2pt) node[anchor=south] {$v_2$} --
(3,0) circle (2pt) node[anchor=south west] {$v_4$};
\node at (-2.5,-.5) {$G$};
\node at (2.5,-.5) {$H$};
\end{tikzpicture}
\caption{For the geometric graph $G,H\in\G^O(\R^2)$, the GMD is zero. The
optimal flow is given by the matching $\pi(u_1)=v_2$, $\pi(u_2)=v_1$,
$\pi(u_3)=v_4$, $\pi(u_4)=v_3$, and $\pi(u_5)=v_5$.}
\label{fig:separability}
\end{figure}
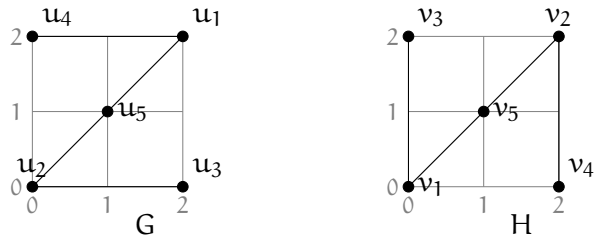
    
\subsection{Metric Properties}
We can see that the GMD induces a pseudo-metric on the space of ordered
geometric graphs $\G^O(\R^d)$. Non-negativity, symmetry, and triangle inequality
follow from those of the cost matrix $[d_{i,j}]$ defined in the GMD.

In addition, we note that $G=H$ (as ordered graphs) implies that $d_{i,j}=0$
whenever $i=j$. The trivial flow, where each $u_i$ sends its full supply to
$v_i$, has a zero cost. So, $\gmd(G,H)=0$. The GMD does not, however, satisfy
the separability condition on $\G^O(\R^d)$.

For the graphs $G,H$ shown in \figref{separability}, we have $\gmd(G,H)=0$. We
note that $G,H$ have the following adjacency length matrices $[E^G_{i}]_i$ and
$[E^H_{j}]_j$, respectively:
\[
\begin{bmatrix}
0 & 0 & 0 &  2& \sqrt{2}\\
0 & 0 & 2 &  0& \sqrt{2}\\
0 & 2 & 0 &  0& 0\\
2 & 0 & 0 &  0& 0\\
\sqrt{2} & \sqrt{2} & 0 &  0& 0
\end{bmatrix}\text{ and }
\begin{bmatrix}
0 & 0 & 2 &  0& \sqrt{2}\\
0 & 0 & 0 &  2& \sqrt{2}\\
2 & 0 & 0 &  0& 0\\
0 & 2 & 0 &  0& 0\\
\sqrt{2} & \sqrt{2} & 0 &  0& 0
\end{bmatrix}\;.
\]
It can be easily checked that the flow that transports a unit of supply from
$u_1\mapsto v_2$, $u_2\mapsto v_1$, $u_3\mapsto v_4$, $u_4\mapsto v_3$,
$u_5\mapsto v_5$, and five units from $u_6\mapsto v_6$ has total cost zero. So,
$\gmd(G,H)=0$. However, the graphs $G$ and $H$ are not the same geometric graph.
The fact that $\ggd(G,H)\neq0$ implies the GGD is not stable under the GMD.

One can easily find even simpler configurations for two distinct geometric
graphs with a zero GMD---if the graphs are allowed to have multiple connected
components.

We conclude this section by stating a stability result for the GMD under the
Hausdorff distance. We omit the proof, since it uses a similar argument
presented in \thmref{ggd}. 
\begin{theorem}[Hausdorff Stability of GMD]
Let $G,H\in\G^O(\R^d)$ be ordered geometric graphs with a bijection $\pi:V^G\to
V^H$ such that $e^G_{i,j}=e^H_{\pi(i),\pi(j)}$ for all $i,j$. If $\delta>0$ is
such that $|u_i-\pi(u_i)|\leq\delta$ for all $u_i\in V^G$, then
\[ \gmd(G,H)\leq C_V|V^G|\delta. \]
\end{theorem}

\subsection{Computing the GMD}
As pointed out earlier, the GMD can be computed as an instance of transportation
problem---using, for example, the network simplex algorithm. If the graphs have
at most $n$ vertices, computing the ground cost matrix $[d_{i,j}]$ takes
$O(n^3)$-time. Since the bipartite network has $O(n)$ vertices and $O(n^2)$
edges, the simplex algorithm runs with a time complexity of $O(n^3)$, with a
pretty good constant. Overall, the time complexity of the GMD is $O(n^3)$. 

\section{Experimental Results}\label{sec:exp} We have implemented the GMD in
Python, using network simplex algorithm from the {\tt networkx} package. We ran
a pattern retrieval experiment on letter drawings from the IAM Graph Database
\cite{da_vitoria_lobo_iam_2008}. The repository provides an extensive collection
of graphs, both geometric and labeled. 

In particular, we performed our experiment on the {\tt LETTER} database from the
repository. The graphs in the database represent distorted letter drawings. The
database considers only $15$ uppercase letters from the English alphabet: {\tt
A}, {\tt E}, {\tt F}, {\tt H}, {\tt I}, {\tt K}, {\tt L}, {\tt M}, {\tt N}, {\tt
T}, {\tt V}, {\tt W}, {\tt X}, {\tt Y}, and {\tt Z}. For each letter, a
prototype line drawing has been manually constructed. On the prototypes,
distortions are applied with three different level of strengths: {\tt LOW}, {\tt
MED}, and {\tt HIGH}, in order to produce $2250$ letter graphs for each level.
Each test letter drawing is a graph with straight-line edges; each node is
labeled with its two-dimensional coordinates. Since some of the graphs in the
dataset were not embedded, we had to compute the intersections of the
intersecting edges and label them as nodes. The preprocessing guaranteed that
all the considered graphs were geometric; a prototype and a distorted graph are
shown in \figref{letter}.
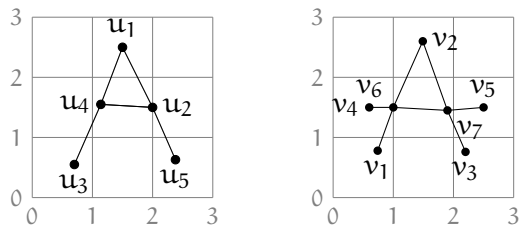
\begin{figure}[tbh]
\centering 
\begin{tikzpicture}[scale=0.8]
        \foreach \i in {0,...,3} { \draw [very thin, gray] (\i-4,0) -- (\i-4,3)  node
        [below] at (\i-4,0) {\footnotesize$\i$}; } \foreach \i in {0,...,3} { \draw
        [very thin, gray] (-4,\i) -- (-1,\i) node [left] at (-4,\i) {\footnotesize$\i$};
        } 
\filldraw (-4+0.7,0.55) circle (2pt) node[anchor=north] {$u_3$} -- (-4+1.14,1.55)
circle (2pt) node[anchor=east] {$u_4$} -- (-4+1.5,0+2.5) circle (2pt)
node[anchor=south] {$u_1$} -- (-4+2,0+1.5) circle (2pt) node[anchor=west] {$u_2$} -- (-4+2.38,0.63)
circle (2pt) node[anchor=north] {$u_5$}; 
\filldraw (-4+2,0+1.5) --  (-4+1.14,1.55);

\foreach \i in {0,...,3} { \draw [very thin, gray] (\i+1,0) --
        (\i+1,3)  node [below] at (\i+1,0) {\footnotesize$\i$}; }
        \foreach \i in {0,...,3} { \draw [very thin, gray] (1,\i) -- (4,\i) node [left]
        at (1,\i) {\footnotesize$\i$}; }

\fill (1+.74,0.78) circle (2pt) node[anchor=north] {$v_1$};
\fill (1+1.49,2.6) circle (2pt) node[anchor=west] {$v_2$};
\fill (1+2.2,0.76) circle (2pt) node[anchor=north] {$v_3$};
\fill (1+.6,1.5) circle (2pt) node[anchor=east] {$v_4$};
\fill (1+2.5,1.5) circle (2pt) node[anchor=south] {$v_5$};
\fill (1+1,1.5) circle (2pt) node[anchor=south east] {$v_6$};
\fill (1+1.9,1.45) circle (2pt) node[anchor=north west] {$v_7$};
\draw (1+.74,0.78) -- (1+1,1.5);
\draw (1+1.49,2.6) -- (1+1,1.5); 
\draw (1+1.49,2.6) -- (1+1.9,1.45);
\draw (1+2.2,0.76) -- (1+1.9,1.45);
\draw (1+.6,1.5) -- (1+1,1.5);
\draw (1+1,1.5) -- (1+1.9,1.45);
\draw (1+2.5,1.5) -- (1+1.9,1.45); 
\end{tikzpicture}
\caption{The prototype geometric graph of the letter {\tt A} is shown on the
left. On the right, a ({\tt MED}) distorted letter {\tt A} is shown.}
\label{fig:letter}        
\end{figure}

We devised a classifier for these letter drawings using the GMD. For this
application, we chose $C_V=4.5$ and $C_E=1$. For a test letter, we computed its
GMD from the $15$ prototypes, then sorted the prototypes in an increasing order
of their distance to the test graph. We then check if the letter generating the
test graph is among the first $k$ prototypes. For each level of distortion and
various values of $k$, we present the rate at which the correct letter has been
found in the first $k$ models. The summary of the empirical results have been
shown in \tabref{result}. Although the graph edit distance based $k$-NN
classifier still outperforms the GMD by a very small margin, our results has
been extremely satisfactory.

\renewcommand{\arraystretch}{1.3}
\begin{table}
\centering
\begin{tabular}{  |m{5em} | m{4.5em} | m{4.5em} | m{4em} |} 
    \cline{2-4}
\multicolumn{1}{c|}{} & \multicolumn{3}{c|}{correct letter in first $k$ models
($\%$)} \\
    \hline
    Distortion & $k=1$ & $k=3$ & $k=5$ \\
    \hline
    {\tt LOW} & $96.66\%$ & $98.93\%$ & $99.37\%$\\
    \hline
    {\tt MED} & $66.66\%$ & $85.37\%$ & $91.15\%$ \\
    \hline
    {\tt HIGH} & $73.73\%$ & $90.48\%$ & $95.51\%$ \\
    \hline
\end{tabular}
\caption{Empirical result on the {\tt LETTER} dataset}
\label{table:result}
\end{table}
One possible reason why the GMD might fail to correctly classify some of the
graphs is that lacks the separability property as a metric. 

\section{Discussions}
We have successfully introduced an efficiently computable and meaningful
similarity measure for geometric graphs. However, the GMD lacks some of the
desirable properties, like separability and stability. The currently presented
stability results for the GGD and GMD have a factor that depends on the size of
the input graphs. The question remains if the distance measures are in fact
stable under much weaker conditions, possibly with constant factors on the right
side. It will also be interesting to study the exact class of geometric graphs
for which the GMD is, in fact, a metric.

\small
\bibliographystyle{abbrv}
\bibliography{references}
\end{document}